\documentclass[conference]{IEEEtran}
\IEEEoverridecommandlockouts

\usepackage[mode=tex]{standalone} 
\usepackage{subcaption}
\usepackage{balance}
\usepackage{cite}
\usepackage[ruled,vlined,linesnumbered]{algorithm2e}

\usepackage{amsmath,amssymb,amsfonts,bbm,bm,mathrsfs}
\usepackage{color}
\usepackage{algorithmic}
\usepackage{graphicx}
\usepackage{textcomp}
\usepackage{xcolor}
\usepackage{amsthm}

\usepackage[ruled,vlined,linesnumbered]{algorithm2e}

\usepackage{tikz}

\usepackage{chronosys}
\usepackage{xspace}
\usetikzlibrary{shapes,arrows,calc,positioning,fit,automata}
\usetikzlibrary{decorations.markings}
\usetikzlibrary{overlay-beamer-styles}
\usepackage{fontawesome}

\usepackage{pgfplots}
\usepackage{pgfplotstable}
\usepackage{filecontents}

\usepackage{caption,subcaption}

\newcommand{\Xcal}{\mathcal{X}}

\newcommand{\Mcal}{\mathcal{M}}
\newcommand{\Gcal}{\mathcal{G}}
\newcommand{\Phat}{\hat{P}}
\newcommand{\Xhat}{\hat{X}}
\newcommand{\xhat}{\hat{x}}

\newcommand{\encF}{\mathrm{f}}
\newcommand{\decC}{\mathrm{c}}

\newcommand{\funID}{\mathrm{id}}
\newcommand{\classL}{\mathrm{L}}

\DeclareMathOperator*{\argmax}{arg\,max}
\DeclareMathOperator*{\argmin}{arg\,min}

\newcommand{\myBlue}{blue!80!black}
\newcommand{\myGreen}{green!60!black}
\newcommand{\myRed}{red!80!black}

\definecolor{NYUviolet}{HTML}{57068c} 	
\definecolor{NYUlight}{HTML}{8900e1} 	
\definecolor{NYUdark}{HTML}{330662} 	
\definecolor{NYUnight}{HTML}{220337} 	

\tikzstyle{block}=[rectangle,draw,very thick,fill=white,align=center]
\tikzstyle{edge} = [draw,very thick,->,-triangle 45]
\tikzstyle{plateBlue} = [draw=\myBlue, shape=rectangle, rounded corners=0.5ex, ultra thick,
minimum width=3.1cm, text=\myBlue, text width=3.1cm, align=right, inner sep=10pt, inner ysep=10pt,append after command={node[font=\bf\Large,text=\myBlue,above left= 3pt of \tikzlastnode.south east] {#1}}]

\tikzstyle{plateRed} = [draw=\myRed, shape=rectangle, rounded corners=0.5ex, ultra thick,
minimum width=3.1cm, text=\myRed, text width=3.1cm, align=right, inner sep=10pt, inner ysep=10pt,append after command={node[font=\bf\Large,text=\myRed,above left= 3pt of \tikzlastnode.south east] {#1}}]

\tikzstyle{plateGreen} = [draw=\myGreen, shape=rectangle, rounded corners=0.5ex, ultra thick,
minimum width=3.1cm, text=green!80!black, text width=3.1cm, align=right, inner sep=10pt, inner ysep=10pt,append after command={node[font=\bf\Large,text=\myGreen,above left= 3pt of \tikzlastnode.south east] {#1}}]

\tikzstyle{plate} = [draw, shape=rectangle, rounded corners=0.5ex, ultra thick,
minimum width=3.1cm,  text width=3.1cm, align=right, inner sep=10pt, inner ysep=10pt,append after command={node[font=\bf\Large,above left= 3pt of \tikzlastnode.south east] {#1}}]

\tikzstyle{plateSmall} = [draw, shape=rectangle, rounded corners=0.5ex, ultra thick,
minimum width=3.1cm,  text width=3.1cm, align=right, inner sep=10pt, inner ysep=10pt,append after command={node[font=\bf,above left= 3pt of \tikzlastnode.south east] {#1}}]

\usepackage[prefix=]{xcolor-material}
\pgfmathdeclarerandomlist{material}{%
	{Red}{Blue}{Green}}
\tikzset{%
	half clip/.code={
		\clip (0, -256) rectangle (256, 256);
	},
	color/.code=\colorlet{fill color}{#1},
	color alias/.code args={#1 as #2}{\colorlet{#1}{#2}},
	colors alias/.style={color alias/.list/.expanded={#1}},
	execute/.code={#1},
	on left/.style={.. on left/.style={#1}},
	on right/.style={.. on right/.style={#1}},
	split/.style args={#1 and #2}{
		on left ={color alias=fill color as #1},
		on right={color alias=fill color as #2, half clip}
	}
}
\newcommand\reflect[2][]{%
	\begin{scope}[#1]\foreach \side in {-1, 1}{\begin{scope}
				\ifnum\side=-1 \tikzset{.. on left/.try}\else\tikzset{.. on right/.try}\fi
				\begin{scope}[xscale=\side]#2\end{scope}
\end{scope}}\end{scope}}

\tikzset{%
	cat/.pic={
		\tikzset{x=1.5cm/5,y=1.5cm/5,shift={(0,-1/3)}}
		\useasboundingbox (-1,-1) (1,2);
		\fill [BlueGrey900] (0,-2)
		.. controls ++(180:3) and ++(0:5/4) .. (-2,0)
		arc (270:90:1/5)
		.. controls ++(0:2) and ++(180:11/4) .. (0,-2+2/5);
		\foreach \i in {-1,1}
		\scoped[shift={(1/2*\i,9/4)}, rotate=45*\i]{
			\clip [overlay] (0, 5/9) ellipse [radius=8/9];
			\clip [overlay] (0,-5/9) ellipse [radius=8/9];
			\fill [BlueGrey900] ellipse [radius=1];
			\clip [overlay] (0, 7/9) ellipse [radius=10/11];
			\clip [overlay] (0,-7/9) ellipse [radius=10/11];
			\fill [Purple100] ellipse [radius=1];
		};
		\fill [BlueGrey900] ellipse [x radius=3/4, y radius=2];
		\fill [BlueGrey100] ellipse [x radius=1/3, y radius=1];
		\fill [BlueGrey900]
		(0,15/8) ellipse [x radius=1, y radius=5/6]
		(0, 8/6) ellipse [x radius=1/2, y radius=1/2]
		{[shift={(-1/2,-2)}, rotate= 10]  ellipse [x radius=1/3, y radius=5/4]}
		{[shift={( 1/2,-2)}, rotate=-10] ellipse [x radius=1/3, y radius=5/4]};
		\fill [BlueGrey500]
		(-1/9,11/8) ellipse [x radius=1/5, y radius=1/5]
		( 1/9,11/8) ellipse [x radius=1/5, y radius=1/5];
		\fill [Purple100]
		(0,12/8)     ellipse [x radius=1/10, y radius=1/5]
		(0,12/8+1/9) ellipse [x radius=1/5 , y radius=1/10];
		\foreach \i in {-1,1}
		\scoped[shift={(1/2*\i,2)}, rotate=35*\i]{
			\clip [overlay] (0, 1/7) ellipse [radius=2/7];
			\clip [overlay] (0,-1/7) ellipse [radius=2/7];
			\fill [Yellow50] ellipse [radius=1];
		};
		\scoped{
			\clip (-1,-2) rectangle ++(2,1);
			\fill [BlueGrey900] (0,-2) ellipse [radius=1/2];
			\fill [Grey100]
			(-1/2,-2) ellipse [x radius=1/3, y radius=1/4]
			( 1/2,-2) ellipse [x radius=1/3, y radius=1/4];
		};
		\foreach \i in {-1,1}
		\foreach \j in {-1,0,1}
		\fill [Grey100, shift={(0,11/8)}, xscale=\i, rotate=\j*15,
		shift=(0:1/2)]
		ellipse [x radius=1/3, y radius=1/64];
	},
dog/.pic={
	\begin{scope}[x=1.5cm/480,y=1.5cm/480]
		\useasboundingbox (-256, -256) (256, 256);
		\reflect[split=Brown400 and Brown500]{
			\fill [fill color] (0,-64) ellipse [x radius=160, y radius=144];
			\fill [fill color] (0, 32) ellipse [x radius=128, y radius=112];
			\fill [fill color] (32,96)
			.. controls ++( 75:128) and ++(105:128) .. ++(192,  0)
			.. controls ++(285: 96) and ++(285: 96) .. ++(-80,-32)
			.. controls ++(105: 64) and ++( 75: 32) .. cycle;
		}
		\reflect[split={Grey100 and Grey200}]{
			\clip (0,-64) ellipse [x radius=160, y radius=144];
			\fill [fill color](0,-224) 
			.. controls ++(  0:64) and ++(270:64) .. ++(96,64)
			.. controls ++( 90:64) and ++(270:64) .. ++(-96,192)
			.. controls ++(270:64) and ++( 90:64) .. ++(-96,-192)
			.. controls ++(270:64) and ++(180:64) .. cycle;
		}
		\reflect[split={Pink100 and Pink200}]{
			\fill [fill color](0,-192) ellipse [x radius=28, y radius=32];
		}
		\reflect[split={BlueGrey800 and BlueGrey900}]{
			\fill [fill color](0,-144) 
			.. controls ++(  0:20) and ++(315:20) .. ++( 40,64)
			.. controls ++(135:20) and ++( 45:20) .. ++(-80, 0)
			.. controls ++(225:20) and ++(180:20) .. cycle;
			\fill [BlueGrey900] (56, 0) circle [radius=20];
			\fill [fill color] (-8,-112)
			-- ++(16,0) -- ++(0,-32) arc (180:360:24)
			arc (180:0:8) arc (360:180:40);
		}
\end{scope}}
}

\tikzset{
	o/.style={
		shorten >=#1,
		decoration={
			markings,
			mark={
				at position 1
				with {
					\draw circle [radius=#1];
				}
			}
		},
		postaction=decorate
	},
	o/.default=2pt
}

\def\BibTeX{{\rm B\kern-.05em{\sc i\kern-.025em b}\kern-.08em
    T\kern-.1667em\lower.7ex\hbox{E}\kern-.125emX}}

\newtheorem{lemma}{Lemma}

\newtheorem{proposition}{Proposition}

\newenvironment{proofsketch}{%
  \proof}{\endproof}

\begin{document}

\title{Single-Shot Compression for Hypothesis Testing
	\thanks{This work was supported in part by NSF--Intel grant \#2003182 and NSF grant \#1925079.}
}

\author{
    \IEEEauthorblockN{Fabrizio Carpi, Siddharth Garg, Elza Erkip}
    \IEEEauthorblockA{Department of Electrical and Computer Engineering, New York University, Brooklyn, NY}
    \{fabrizio.carpi, siddharth.garg, elza\}@nyu.edu
}

\maketitle

\begin{abstract}

Enhanced processing power in the cloud allows constrained devices to offload costly computations: for instance, complex data analytics tasks can be computed by remote servers.
Remote execution calls for a new compression paradigm that optimizes performance on the analytics task within a rate constraint, instead of the traditional rate-distortion framework which focuses on source reconstruction.
This paper considers a simple binary hypothesis testing scenario where the resource constrained client (transmitter) performs fixed-length single-shot compression on data sampled from one of two distributions; the server (receiver) performs a hypothesis test on multiple received samples to determine the correct source distribution. 
To this end, the task-aware compression problem is formulated as finding the optimal source coder that maximizes the asymptotic error performance of the hypothesis test on the server side under a rate constraint. 
A new source coding strategy based on a greedy optimization procedure is proposed and it is shown that that the proposed compression scheme outperforms universal fixed-length single-shot coding scheme for a range of rate constraints. 

\end{abstract}


\section{Introduction}

Access to higher bandwidth and lower latency wireless technology is accelerating the use of edge computing. 
In edge computing, a resource constrained \emph{client}, a mobile phone or a sensor for example, outsources computations to a remote \emph{server} over a wireless link. 
Typically, the computations involve decision and analytics tasks over the transmitted data: for instance, image classification, object detection or speech recognition.
For efficient bandwidth usage, the client might seek to compress the source data before transmitting to the server. 
However, traditional compression (or source coding) schemes are optimized for source reconstruction, that is, the seek to minimize a distortion metric (e.g., mean squared error) between the transmitted and the received data. 
Nonetheless, distortion does not directly correspond to the receiver's goal in the edge computing scenario. 
In this case, the receiver's goal is to maximize performance on the analytics tasks of interest. 
This gives rise to the central question of this paper: how can we design \emph{task-aware} source coding schemes which provide \emph{effective} representations of the source data so as to successfully carry out the analytics task?


One answer to this question is to use a distortion metric that is tailored for common analytics tasks. 
Motivated by this idea, recent works~\cite{Shkel-18-TIT, Shkel-17-ISIT} have studied the rate-distortion tradeoffs for the logarithmic loss distortion measure, since log-loss is commonly used in the machine learning community in the context of classification tasks. 
However, even log-loss distortion measure is ultimately a proxy for the analytics task at hand. 
How much better could one do by tailoring the compression scheme for the exact analytics task?

In this paper, we investigate task-aware compression for a simple edge computing scenario.
We select binary hypothesis testing as a candidate task since it is both commonplace and well understood mathematically. 
In binary hypothesis testing the source data is sampled from one of two distributions and the goal is to decide which one was the correct source distribution.

Next, we model the client's resource constraints --- an unconstrained client could perform the hypothesis test by itself and transmit a single bit (binary decision) to the server.
In contrast, our primary assumption is that the client does not have processing capabilities to compute the task locally.
We model a resource-constrained client that only has sufficient resources to store and process a \emph{single} data sample at a time; as such, it compresses each data sample it receives using a simple \emph{scalar} compression scheme (as opposed to vector compression) and transmits to the server, over a rate-limited link.
In literature, this is referred to as ``single-shot'' compression. 
We assume fixed-length (lossy) compression, i.e., the compressed samples belong to an alphabet with size limited by the rate constraint. 
The server, on the other hand, is computationally unconstrained and collects an arbitrarily large number of compressed samples from the client for hypothesis testing.

Versions of this problem have been investigated in a multi-terminal setting with compression over large blocklengths~\cite{Han-Amari-98}. 
In most of this literature, no resource constraints are assumed on the clients and the asymptotic performance is provided.
Ziv~\cite{Ziv-88} analyzes 
binary hypothesis testing with empirically observed statistics; a link to universal compression is established but applies only over large blocklengths, while we are interested in single-shot compression.
Prior work has also looked at the related problem of learning classification-oriented compressed data 
representations~\cite{Li-20}, where both the client and server operate on a single sample of data as it is customary in classification settings, as opposed to hypothesis testing that operates over large blocklengths~\cite{HypTest-Classif}. 

The main focus of this paper is to design an effective task-aware source coder for binary hypothesis testing.
In Section~\ref{sec:system-model}, we start by formally defining the system model, where we take into account the client constraints mentioned above.
In Section~\ref{sec:compress-hyptest}, we formalize the fixed-length single-shot compression for hypothesis testing problem; we also define the optimal compressor, which requires exponential (in the alphabet size) complexity for the construction.
Then, we propose a task-oriented compression scheme in Section~\ref{sec:proposed-scheme}: our scheme is based on a greedy optimization which aims to the preserve the \emph{useful} information between the two source hypotheses, in this case the Kullback-Leibler distance between the two distributions.
The proposed compressor is constructed through iterative steps and it can be determined in polynomial time.
In Section~\ref{sec:results}, we show empirical results and computational bounds for our compressor. 
Finally, our conclusions are discussed in Section~\ref{sec:conclusion}.

\section{System Model}
\label{sec:system-model}

\begin{figure}
    \centering
    	\begin{tikzpicture}
		\def\dhor{4}
		\def\dvert{2.5}
		\def\dhorinit{3}
	
		\node[block, draw=white, align=center] (P) at (-\dhorinit,\dvert) {$H_0: X\sim P_0$\\$H_1: X\sim P_1$};
		\node[block, draw=white, align=center] (X) at (-\dhorinit,0) {$X$};
		\node[block, inner sep=5pt, align=center, draw] (f) at (0,0) {One-Shot\\Compressor\\$\encF(\cdot)$};
		\node[block, inner sep=5pt, align=center, draw] (b) at (\dhor,0) {Buffer\\$n$};
		\node[block, inner sep=5pt, align=center, draw] (c) at (\dhor+\dhorinit,0) {Hypot.\\Test.\\$\classL(\cdot)$};
		\node[block, draw=white, align=center] (Phat) at (\dhor+\dhorinit,\dvert) {Accept or\\Reject $H_0$};
		
		\node[plateSmall=TX-Client, fit=(f), inner sep=15pt,minimum height=80] (client) at (0,0) {};
		\node[plateSmall=RX-Server, fit=(c)(b), inner sep=15pt,minimum height=80] (server) at (\dhor+0.5*+\dhorinit,0) {};

		\path (P) edge[dotted, ultra thick, -triangle 45]  (X);
		\path (X) edge[edge, -triangle 45]  (f);
		\path (f) edge[edge, -triangle 45] node[above]{$\Xhat$} (b);
		\path (b) edge[edge, -triangle 45] node[above]{$\Xhat^n$} (c);
		\draw[edge, -triangle 45] (c) -- (Phat);
	\end{tikzpicture}
    \caption{System Model.}
    \label{fig:system-model}
\end{figure}
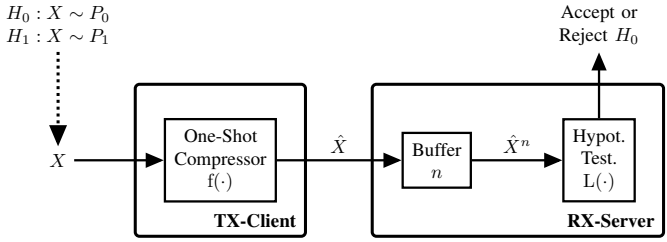

The system model is shown in Fig.~\ref{fig:system-model}.
Throughout the rest of the paper the client is called transmitter and the server is called receiver.
The data comes from one of the two distributions $P_\theta$, $\theta\in\{0,1\}$, where $\theta=0$ represents the null hypothesis $H_0$ and $\theta=1$ represents the alternative hypothesis $H_1$.
We have $X_1,\dots,X_n\sim P_\theta$ i.i.d. random variables defined over a finite alphabet $\Xcal=\{1,\dots,|\Xcal|\}$.
The transmitter, due to memory constraints, cannot store and process $X^n$ jointly to do hypothesis testing.
Instead, it sends the one-shot (scalar) compressed $X^n$ to the receiver where hypothesis testing takes place.

Formally, at the transmitter, the single-shot compressor $\encF$ is a surjective function defined as
\begin{equation}
    \encF:\:\Xcal\to\Mcal
\end{equation}
where $\Mcal=\{1,\dots,M\}$ is the compressed alphabet of size $M$.
We denote $\Xhat=\encF(X)$, i.e., $\Xhat$ represents the mapping of the source letter $X$.
We consider $M<|\Xcal|$, since for $M\geq |\Xcal|$ there is no need for compression.
This corresponds to fixed rate compression with rate $R=\log M$.\footnote{Throughout this paper $\log(\cdot)$ is assumed to be base 2.}


The probability distribution of $\Xhat$ under $P_\theta$, $\theta\in\{0,1\}$, is denoted as $\Phat_\theta$ and is given by
\begin{equation}
    \Phat_\theta(\xhat) = \sum_{x:\encF(x)=\xhat} P_\theta(x).
    \label{eq-compr-distr}
\end{equation}



The receiver observes $\Xhat_1,\dots,\Xhat_n$ and either accepts or rejects the null hypothesis.
Using standards definitions in simple hypothesis testing~\cite{CoverThomas-IT}, type-I error, denoted as $\alpha_n$, occurs when the null hypothesis ($\theta=0$) is true, but the receiver rejects it.
Instead, type-II error, denoted as $\beta_n$, corresponds to the receiver accepting the null hypothesis when the alternative hypothesis ($\theta=1$) is true.
It is known that in the classical hypothesis testing setting, for any $\epsilon\in(0,1/2)$ and $\alpha_n<\epsilon$, the optimal type-II error $\beta_n^\epsilon$ decays exponentially in $n$ with exponent $\gamma$ defined as
\begin{equation}
    \gamma = - \lim_{n\to\infty} \frac{1}{n} \log \beta^\epsilon_n. 
	\label{eq:chernoff-stein}
\end{equation}
We say that $(R,\eta)$ is \emph{achievable} if there exists a single-shot  rate $R$ compressor at the client and a corresponding hypothesis testing function at the server with type-I error less  than $\epsilon$ and type-II error exponent $\eta$.
Note that type-II error exponent does not typically depend on type-I error bound $\epsilon$~\cite{CoverThomas-IT} as long as $\epsilon$ is fixed, hence we will not explicitly state the dependency on $\epsilon$.
In particular, for a given compression rate $R$, we would like to find the largest achievable type-II error exponent
\begin{equation}
    \gamma^{\star} (R) = \sup\{\eta:\:(R,\eta)\text{ achievable}\}.
\end{equation}
Note that if $R=\log(|\Xcal|)$ and the compressor is the identity transformation $\funID(\cdot)$, then Chernoff-Stein lemma~\cite{CoverThomas-IT} determines the optimal error exponent
\begin{equation}
    \gamma^\star(\log|\Xcal|) = \gamma_\funID(\log|\Xcal|) = D(P_0||P_1),
\end{equation}
where $D(P_0||P_1)$ is the Kullback--Leibler (KL) divergence between $P_0$ and $P_1$~\cite{CoverThomas-IT}.
The error exponent penalty for a rate $R$ compressor $\encF$ at is defined as
\begin{equation}
    \Delta_\encF (R) =  D(P_0||P_1) - \gamma_{\encF} (R),
\end{equation}
where $\gamma_{\encF}(R)$ is the largest type-II error exponent determined by the compressor $\encF$.
The optimal penalty is
\begin{equation}
    \Delta^{\star} (R) = D(P_0||P_1) - \gamma^\star (R).
\end{equation}

\section{Hypothesis Testing under Single-Shot Compression}
\label{sec:compress-hyptest}

For the  one-shot compressed  binary  hypothesis testing problem, our first result states that the log-likelihood ratio (LLR) test using  the compressed variables $\Xhat_1,\dots,\Xhat_n$ is optimal.

\begin{lemma}[Hypothesis testing on compressed variables]
\label{lem:NP-compr}
The following LLR test on compressed variables $\Xhat_i=\encF(X_i)$, $i=1,\dots,n$, is optimal.
\begin{align}
	L(\Xhat_1,\dots,\Xhat_n) = 
	\sum_{i=1}^{n} \log \frac{\Phat_0(\Xhat_i)}{\Phat_1(\Xhat_i)} 
	\overset{{\hat\theta=0}}{\underset{{\hat\theta=1}}{\gtrless}} \log T,
	\label{eq:LLR-test}
\end{align}
where $T\geq0$ depends on the type-I error exponent bound $\epsilon$.
The corresponding optimal error exponent is 
\begin{equation}
    \gamma_\encF(R) = D(\Phat_0||\Phat_1).
\end{equation}

\end{lemma}
\begin{proofsketch}
    Since the source random variable is i.i.d. and the compressor function is $\encF$ memoryless, the compressed variable is also i.i.d. $\Xhat_1,\dots,\Xhat_n\sim \Phat_\theta$. 
    Then, Neyman-Pearson test~\cite[Chapter 11]{CoverThomas-IT} can be applied to $\Xhat^n$.
    Moreover, Chernoff-Stein lemma determines that the the optimal error exponent is equal to the KL divergence between the distribution of the compressed variables under the two hypotheses.
\end{proofsketch}

As discussed in Section~\ref{sec:system-model}, the error exponent $\gamma_\encF(R)$ determines the speed of convergence 
--- intuitively, the farther apart the two compressed distributions (large KL divergence), the faster the type-II error probability goes to zero.
Hence, our goal is to find a compressor $\encF$ which induces a partition of $M$ sets over $\Xcal$ such that the KL distances between the compressed distributions $D(\Phat_0||\Phat_1)$ is maximized.
Clearly, compression reduces the error exponent (we will mathematically show this in  Proposition~\ref{thm:grouping-KL}) and by Lemma~\ref{lem:NP-compr} the smallest compression penalty for the compressor $\encF$ is
\begin{equation}
    \Delta_{\encF} (R) = D(P_0||P_1) - D(\Phat_0||\Phat_1).
    \label{eq:compression-penalty}
\end{equation}
Then, the optimal compressor $\encF^\star$ at rate $R=\log M$ is 
\begin{equation}
	\encF^\star = \argmax_\encF D(\Phat_0||\Phat_1) \quad \text{ s.t. }\:|\encF|\leq M, \label{eq:optimiz-prob}
\end{equation}
or, equivalently,
\begin{equation}
    \encF^\star = \argmin_\encF \Delta_\encF(R) \quad \text{ s.t. }\:|\encF|\leq M.
    \label{eq:optimal-compressor}
\end{equation}
where $|\encF|$ is the cardinality of the compression function. 

In the following proposition we derive a useful analytical expression for $\Delta_\encF(R)$ in terms of distributions over compressed symbols.
For mathematical convenience, we  define $\Gcal_{\xhat} = \{x: \encF(x)=\xhat \}$; this set includes the source outcomes mapped to the compressed symbol $\xhat$.
Hence, the compressor induces the ``groups'' $\Gcal_{\xhat}\in\{ \Gcal_1,\dots,\Gcal_M\}=\Gcal$, which form a partition over $\Xcal$.
    
\begin{proposition}[Compression penalty on type-II error exponent]
\label{thm:grouping-KL}
    For any compressor $\encF$, the minimal compression penalty is $\Delta_\encF(R)\geq 0$ and can be expressed as:
    \begin{equation}
        \Delta_{\encF}(R) = \sum_{\xhat=1}^{M} \Phat_0(\xhat)\,D\Big(P_0(x|\xhat)\Big|\Big|P_1(x|\xhat)\Big)
        \label{eq:grouping-penalty}
    \end{equation}
    where the posterior distribution of $X$ given the compressed realization $\encF(X)=\xhat$ is
    \begin{equation}
        P_{\theta}(x|\xhat) =  \frac{P_\theta(x)}{\Phat_\theta(\xhat)}\mathbbm{1}\{\xhat=\encF(x)\}.
        \label{eq:group-posterior}
    \end{equation}
\end{proposition}
\begin{proof}

    Expanding equation~\eqref{eq:compression-penalty}:
    {\small
    \begin{align}
		&\Delta_{\encF}(R) = \sum_{x\in\Xcal} P_0(x)\log\frac{P_0(x)}{P_1(x)} - \sum_{\xhat\in\Mcal} \Phat_0(\xhat)\log\frac{\Phat_0(\xhat)}{\Phat_1(\xhat)}\nonumber\\
		&=\sum_{\xhat\in\Mcal} \sum_{x\in\Gcal_{\xhat}} P_0(x)\log\frac{P_0(x)}{P_1(x)} - \sum_{\xhat\in\Mcal} \left(\sum_{x\in\Gcal_{\xhat}}P_0(x)\right)\log\frac{\Phat_0(\xhat)}{\Phat_1(\xhat)}
		\label{eq:proof-KL-group-positive}\\
		& = \sum_{\xhat\in\Mcal} \sum_{x\in\Gcal_{\xhat}}P_0(x) \log \left(\frac{P_0(x)}{\Phat_0(\xhat)}\,\frac{\Phat_1(\xhat)}{P_1(x)} \right) \label{eq:proof-KL-2}\\
		&= \sum_{\xhat\in\Mcal} \sum_{x\in\Gcal_{\xhat}}P_0(x) \log \frac{P_0(x|\xhat)}{P_1(x|\xhat)}  \label{eq:proof-KL-3}\\
		& = \sum_{\xhat\in\Mcal} \Phat_0(\xhat)\,D\Big(P_0(x|\xhat)\Big|\Big|P_1(x|\xhat)\Big)\nonumber
	\end{align}
	}
    where: in~\eqref{eq:proof-KL-group-positive} we used the definition~\eqref{eq-compr-distr}; in~\eqref{eq:proof-KL-group-positive} and~\eqref{eq:proof-KL-2} we used the fact that $\Gcal_1,\dots,\Gcal_M$ form a partition over $\Xcal$; in~\eqref{eq:proof-KL-3} we used the definition~\eqref{eq:group-posterior} since $P(\Xhat|X) =  \mathbbm{1}\{\Xhat=\encF(X)\}$.
	Note that if $\Gcal_{\xhat}$ contains a single element (one-to-one mapping), then $\,D\big(P_0(x|\xhat)||P_1(x|\xhat)\big)=0$.
	Moreover~\eqref{eq:proof-KL-group-positive} is greater than zero by the log-sum inequality.
\end{proof}

Non-negativity of $\Delta_\encF(R) \geq 0$ can also be observed from equation~\eqref{eq:grouping-penalty} as it is a convex combination of KL-distances, each individually positive. 
Proposition~\ref{thm:grouping-KL} also yields an important intuition about optimal compression:
note that the $\xhat$-th term in~\eqref{eq:grouping-penalty} is directly proportional to the relative entropy between the posteriors over the $\xhat$-th group $\Gcal_{\xhat}$ induced by $\encF$.
As a consequence, \eqref{eq:grouping-penalty} suggests that a good task-aware compression strategy combines the source letters that have similar posteriors over the compressed groups; in other words, the probability ratios between the combined letters under $P_0$ has to be similar to the ones under $P_1$. 


\section{Proposed Compressor}
\label{sec:proposed-scheme}

\begin{algorithm}[t]
	\caption{KL-greedy compressor's construction}
	\label{alg:compres-hyp-test}
	\SetAlgoLined
	\SetKwInOut{Input}{Input}
	\SetKwInOut{Output}{Output}
	\SetKwProg{Init}{Init: }{}{}
	\SetKwRepeat{Do}{do}{while}
	
	\Input{Source distributions $P_0, P_1$; rate $M$.} 
	Initialize: $\Phat_0\gets P_0$, $\Phat_1\gets P_1$, $\Gcal\gets\{\{1\},\dots,\{|\Xcal|\}\}$.
	

\For{$k=1,\dots,|\Xcal|-M$}{
	Find $\{\Gcal_a,\Gcal_b\}\subset\Mcal_k$ which minimize~\eqref{eq:grouping-one-step}.\\
	Remove the $b$-th entry and combine $\{\Gcal_a,\Gcal_b\}$ by updating the $a$-th entry:\\
	$\Phat_0\gets [\dots,\Phat_0(\Gcal_a)+\Phat_0(\Gcal_b),\dots,0,\dots]$\\
	$\Phat_1\gets [\dots,\Phat_1(\Gcal_a)+\Phat_1(\Gcal_b),\dots,0,\dots]$\\
	$\mathcal G\gets [\dots,\Gcal_a \cup \Gcal_b,\dots,\emptyset,\dots]$
}

	\Output{Compressed distr. $\Phat_0,\Phat_1$; groups $\mathcal{G}$.}
\end{algorithm}

When solving the optimization problem in~\eqref{eq:optimiz-prob}, one has to consider all the possible surjective functions $\encF$ which induce valid partitions over the source alphabet; the number of such number of partitions is exponential in the source/compressed alphabet size.  
Partitioning problems of this nature have been shown to be NP-Hard~\cite[Chapter 3]{book-theory_NP},\cite{Kai-15}.

In this paper, we propose an efficient (i.e., polynomial time) greedy approximation for the optimal compressor.
The following lemma is the basis for our construction.
\begin{lemma}[One-step Compression from $|\Xcal|$ to $|\Xcal|-1$]
	\label{cor:one-step}
	Let $\encF$ be a compression rule which groups two letters $\{a,b\}\subset\Xcal$.
	That is, $\Gcal_{m}=\{a,b\}$, $m\in\Mcal$, and the others groups $\Gcal_i$, $i=1,\dots,M$, $i\neq m$, are one-to-one.
	Then, the optimal compressor for $M=|\Xcal|-1$ induces the groups $\Gcal^\star$, minimizing the compression penalty
	\begin{align}
        \Gcal^\star = &\argmin_{\Gcal_m=\{a,b\}\subset\Xcal} \bigg\{\Phat_0(m) D\Big(P_0(x|m)\Big|\Big|P_1(x|m)\Big) \bigg\},
		\label{eq:grouping-one-step}
	\end{align}
	where the posteriors over the candidate group $\Gcal_m=\{a,b\}$ are simply defined as
    \begin{equation}
        P_\theta(x|m) = \left[\frac{P_\theta(a)}{P_\theta(a)+P_\theta(b)}, \frac{P_\theta(b)}{P_\theta(a)+P_\theta(b)} \right].
    \end{equation}
\end{lemma}
Note that if the groups $\Gcal_i$ are one-to-one, the $i$-th KL divergence term in~\eqref{eq:grouping-penalty} is 0.
Intuitively, when reducing the alphabet size by one, the optimal compressor combines the two letters that minimize the product of the probability of the group and the  KL distance between the posteriors over the group. 

For general $M$, we propose an iterative construction of the compressor that reduces the compressed alphabet size by one in each step.
Denote the steps by $k=1,\dots,|\Xcal|-M$, where $M$ is the target rate.
Let $\Mcal_k$ be the compressed alphabet at the $k$-th step, with size $|\Mcal_k| = |\Xcal|-k$, with $k=1,\dots,|\Xcal|-M$.
Let $\Gcal_1,\dots,\Gcal_{|\Mcal_k|}$ be the corresponding partition on $\Xcal$ at the $k$-th step.
For example, at the first step $k=1$, the (optimal) groups $\Gcal_1,\dots,\Gcal_{|\Xcal|-1}$ are computed according to Lemma~\ref{cor:one-step}.
Generally, at step $k>1$, our compressor combines the two groups $\{\Gcal_a,\Gcal_b\}_k^\star\subset\Mcal_k$ that minimize~\eqref{eq:grouping-one-step}, where $\Xcal$ is replaced by $\Mcal_k$ and $\{\Gcal_a,\Gcal_b\}$ is a generalization of $\{a,b\}$.
Finally, the compression function $\encF$ is defined such that $\encF(x) = \xhat$ if $x\in\Gcal_{\xhat}$.
We call our proposed compressor ``KL-greedy'' and its construction is summarized in Algorithm~\ref{alg:compres-hyp-test}.
Note that the number of pairs of groups $\{\Gcal_a,\Gcal_b\}$ that need to be considered at the $k$-th step is $\binom{|\Mcal_k|}{2}$.
Thus, our compressor can be designed in polynomial time.

\section{Results}
\label{sec:results}

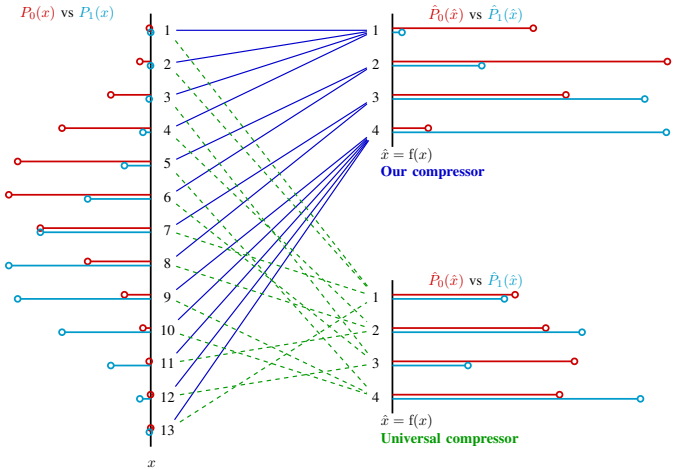
\begin{figure}[t]
    \centering
    \begin{tikzpicture}

\def\dhor{5}
\def\dvert{1}
\def\dnodes{0.8}
\def\dplot{0.4}
\def\dsmall{0.05}

\def\Nin{13}
\def\Nout{4}

\def\scalePin{15}
\def\scalePout{1}
\def\scalePmid{1}
\def\dcompr{8}


\foreach \i in {1,...,\Nin}{
	\node (x-\i) at (0,-\dvert-\i*\dnodes) {\i};
}

\draw[-, very thick] (-\dplot,-\dvert-\dnodes+\dplot) -- (-\dplot,-\dvert-\Nin*\dnodes-\dplot);
\node[anchor=west] at (-1.5*\dplot,-\dvert-\Nin*\dnodes-2*\dplot){$x$};

\node at (-6*\dplot,-\dvert-\dnodes+\dplot) {{\color{\myRed}$P_0(x)$} vs {\color{cyan!80!black}$P_1(x)$}};
\foreach \i/\P/\Q/\R in {
1/0.00218/0.00002/0.00113,
2/0.01741/0.00030/0.00889,
3/0.06385/0.00249/0.03319,
4/0.14189/0.01246/0.07718,
5/0.21284/0.04204/0.12743,
6/0.22703/0.10090/0.16392,
7/0.17658/0.17658/0.17652,
8/0.10090/0.22703/0.16392,
9/0.04204/0.21284/0.12743,
10/0.01246/0.14189/0.07718,
11/0.00249/0.06385/0.03319,
12/0.00030/0.01741/0.00889,
13/0.00002/0.00218/0.00113
}{
	\draw [o, very thick,draw=\myRed] (-\dplot,-\dvert-\i*\dnodes+\dsmall) -- (-\dplot-\scalePin*\P,-\dvert-\i*\dnodes+\dsmall);
	\draw [o, very thick,draw=cyan!80!black] (-\dplot,-\dvert-\i*\dnodes-\dsmall) -- (-\dplot-\scalePin*\Q,-\dvert-\i*\dnodes-\dsmall);
}

\foreach \i in {1,...,\Nout}{
	\node (m-\i) at (\dhor,-\dvert-\i*\dnodes) {\i};
}
\draw[-, very thick] (\dhor+\dplot,-\dvert-\dnodes+\dplot) --node[near end, below=1.2, anchor=west,align=left]{} (\dhor+\dplot,-\dvert-\Nout*\dnodes-\dplot);
\node[anchor=west,align=left] at (\dhor,-\dvert-\Nout*\dnodes-\dnodes){$\xhat=\encF(x)$\\
{\color{\myBlue}\bf Our compressor}};
\node at (\dhor+6*\dplot,-\dvert-\dnodes+\dplot) {{\color{\myRed}$\Phat_0(\xhat)$} vs {\color{cyan!80!black}$\Phat_1(\xhat)$}};
\foreach \i/\P/\Q in {
1/0.22534/0.01527,
2/0.43987/0.14295,
3/0.27748/0.40361,
4/0.05731/0.43818
}{
	\draw [o, very thick,draw=\myRed] (\dhor+\dplot,-\dvert-\i*\dnodes+\dsmall) -- (\dhor+\dplot+\scalePin*\P,-\dvert-\i*\dnodes+\dsmall);
	\draw [o, very thick,draw=cyan!80!black] (\dhor+\dplot,-\dvert-\i*\dnodes-\dsmall) -- (\dhor+\dplot+\scalePin*\Q,-\dvert-\i*\dnodes-\dsmall);
}

\foreach \i/\j in {
1/1,
2/1,
3/1,
4/1,
5/2,
6/2,
7/3,
8/3,
9/4,
10/4,
11/4,
12/4,
13/4}{
	\path (x-\i) edge[thick,\myBlue] (m-\j);
}

\foreach \i in {1,...,\Nout}{
	\node (mq-\i) at (\dhor,-\dvert-\dcompr*\dnodes-\i*\dnodes) {\i};
}
\draw[-, very thick] (\dhor+\dplot,-\dvert-\dcompr*\dnodes-\dplot) --node[near end, below=1.2, anchor=west]{} (\dhor+\dplot,-\dvert-\Nout*\dnodes-\dcompr*\dnodes-\dplot);
\node[anchor=west,align=left] at (\dhor,-\dvert-\dcompr*\dnodes-\Nout*\dnodes-\dnodes){$\xhat=\encF(x)$\\{\color{\myGreen}\bf Universal compressor}};
\node at (\dhor+6*\dplot,-\dcompr*\dnodes-\dvert-\dnodes+\dplot) {{\color{\myRed}$\Phat_0(\xhat)$} vs {\color{cyan!80!black}$\Phat_1(\xhat)$}};
\foreach \i/\P/\Q in {
1/0.19619/0.17907,
2/0.24529/0.30334,
3/0.29118/0.12081,
4/0.26734/0.39678
}{
	\draw [o, very thick,draw=\myRed] (\dhor+\dplot,-\dcompr*\dnodes-\dvert-\i*\dnodes+\dsmall) -- (\dhor+\dplot+\scalePin*\P,-\dcompr*\dnodes-\dvert-\i*\dnodes+\dsmall);
	\draw [o, very thick,draw=cyan!80!black] (\dhor+\dplot,-\dcompr*\dnodes-\dvert-\i*\dnodes-\dsmall) -- (\dhor+\dplot+\scalePin*\Q,-\dcompr*\dnodes-\dvert-\i*\dnodes-\dsmall);
}
\foreach \i/\j in {
1/1,
2/1,
3/3,
4/2,
5/4,
6/3,
7/1,
8/2,
9/4,
10/4,
11/2,
12/3,
13/1}{
	\path (x-\i) edge[dashed,thick,\myGreen] (mq-\j);
}

\end{tikzpicture}
    \caption{
    Left: Source distributions for $|\Xcal|=13$. 
    Top-right: compressed distributions for our compressor of Algorithm~\ref{alg:compres-hyp-test}; the solid blue line shows the mappings of the compression function.
    Bottom-right: compressed distributions for the universal compressor from~\cite{Shkel-17-ISIT}; the dashed green line shows the mappings of the compression function.}
    \label{fig:distributions}
\end{figure}

In this section, we discuss numerical results and performance of Algorithm~\ref{alg:compres-hyp-test}.
We consider $P_\theta$ to be a (shifted) binomial distribution over $\Xcal$ with parameter $s_\theta$, i.e., 
\begin{equation}
    P_\theta(x) = \binom{|\Xcal|-1}{x-1} s_\theta^{x-1} (1-s_\theta)^{|\Xcal|-x}.
\end{equation}

We quantify the compression penalty $\Delta_\encF(R)$ based on~\eqref{eq:compression-penalty}.  
We also estimate type-II error rate by performing the LLR test~\eqref{eq:LLR-test} on the receiver side; we consider blocklength $n=5$ and bound on the type-I error $\epsilon=0.05$.
The threshold $T$ is empirically chosen such that it is the largest value for which the estimated type-I error is $N(\hat\theta=1,\theta=0)/N(\theta=0)<\epsilon$, for a given compressor $\encF$ at rate $M$; $N(\cdot)$ is the counting function.
The type-II error rate is empirically estimated as $N(\hat\theta=0,\theta=1)/N(\theta=1)$. 
Both estimates are computed over $N(\theta=0)=N(\theta=1)=10^6$ realizations of source blocks $x^n$.

\subsection{Baseline: Single-shot Universal Lossy Source Coding under Logarithmic Loss}
\label{sec:background-univ-compression}

Universal compression schemes are designed to perform well over a family of source distributions --- the family $\{P_0,P_1\}$ in our scenario.
In compliance with our system model, we consider the universal fixed-length single-shot lossy compression scheme analyzed by Shkel et al. in~\cite{Shkel-17-ISIT}.
We recall that although this universal compressor is task-unaware, it is designed for soft reconstruction under logarithmic loss distortion, which generally provides ``universally good'' schemes~\cite{No-2019-logloss}.
The construction of this universal compressor aims to find $Q^\star$, a distribution over $\Xcal$ which is used to approximate the source distribution over the family $\{P_0,P_1\}$.
As in~\cite{Shkel-17-ISIT}, for a rate constraint $R=\log M$, $Q^\star$ belongs to 
\begin{equation}
    \mathcal{Q}_M = \{Q: \min_{x\in\Xcal} \log\frac{1}{Q(x)} \geq \log M\}.
\end{equation}
For every value of $M$, $Q^\star$ is the solution of the following optimization problem
\begin{align}
    Q^\star = \argmin_{Q\in\mathcal Q_M} \delta \quad
    \text{ s.t.: } \begin{cases}
     D(P_0 || Q) \leq \delta,\\
     D(P_1 || Q) \leq \delta.
     \end{cases}
\end{align}
In other words, $Q^\star$ can be seen as a distribution that is ``equidistant'' from the two hypotheses.
Given $Q^\star$, the universal compressor is constructed according to~\cite[Theorem 4]{Shkel-18-TIT}.
Intuitively, the letters corresponding to the largest values of $Q^\star$ get one-to-one mappings, while the letters corresponding to the lowest values of $Q^\star$ get grouped together.

\subsection{Simulation Results}

We show the performance of different compressors in our hypothesis testing scenario. 
In the figures, we show empirical results for different compression functions $\encF$:
\begin{itemize}
    \item Uncompressed: no compression is performed, i.e.,  $\xhat=x$;
    \item Optimal compressor: defined in~\eqref{eq:optimal-compressor};
    \item Our KL-greedy compressor: defined in Section~\ref{sec:proposed-scheme} and Algorithm~\ref{alg:compres-hyp-test};
    \item Universal compressor: defined in~\cite{Shkel-17-ISIT} and briefly introduced in Section~\ref{sec:background-univ-compression}.
\end{itemize}

\begin{figure}[t]
    	\centering
        \pgfplotsset{
    compat=newest,
    /pgfplots/legend image code/.code={%
        \draw[mark repeat=2,mark phase=2,#1] 
            plot coordinates {
                (0cm,0cm) 
                (0.6cm,0cm)
                (1.2cm,0cm)%
            };
    },
}

\begin{tikzpicture}
\begin{axis}[
	scale only axis,
	font=\Large,
	width=\textwidth,
	height=0.25\textwidth,
	xmajorgrids=true,
	ymajorgrids=true,
	xmin=1,
	xmax=3.585,
	ymin=0,
	ymax=1.5,
	title={Compression penalty $\Delta_\encF(R)$, $|\Xcal|=13$},
	xlabel={$\log M$},
	legend columns=1,
	legend cell align={left},
	legend style={row sep=3pt, at={(1,1)}, anchor = north east},
	]

    \addplot+[line width=2pt, orange!80!black, mark=x, mark options={scale=4,fill=orange!80!black,solid}] table[x=X, y=Y, col sep=comma] 
    {fig/delta_Xshort_data_opt.txt};
    
    \addplot+[line width=2pt, \myBlue, densely dotted, mark=*, mark options={scale=2,fill=\myBlue,solid}] table[x = X, y=Y, col sep=comma] {fig/delta_Xshort_data_our.txt};
    
    \addplot+[line width=2pt, \myGreen, loosely dashed, mark=square, mark options={scale=2,fill=\myGreen,solid}] table[x=X, y=Y, col sep=comma] 
    {fig/delta_Xshort_data_Q.txt};

    \legend{Optimal compressor, Our compressor,Universal compressor}
\end{axis}
\end{tikzpicture}
    	\caption{
    	Compression penalty for $|\Xcal|=13$.
    	}
    	\label{fig:KL-results-short}
\end{figure}
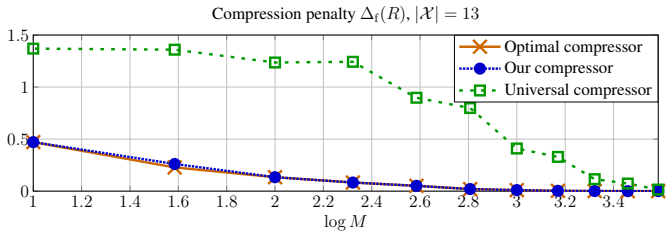

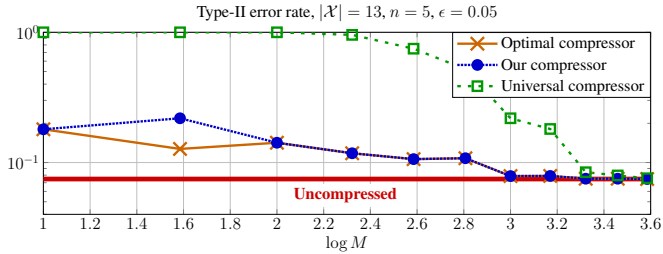
\begin{figure}[t]
        \centering
        \pgfplotsset{
    /pgfplots/legend image code/.code={%
        \draw[mark repeat=2,mark phase=2,#1] 
            plot coordinates {
                (0cm,0cm) 
                (0.6cm,0cm)
                (1.2cm,0cm)%
            };
    },
}

\begin{tikzpicture}
\begin{semilogyaxis}[
	scale only axis,
	font=\Large,
	width=\textwidth,
	height=0.3\textwidth,
	xmajorgrids=true,
	ymajorgrids=true,
	xmin=1,
	xmax=3.6,
	ymin=0.04,
	ymax=1,
	title={Type-II error rate, $|\Xcal|=13$, $n=5$, $\epsilon=0.05$},
	xlabel={$\log M$},
	legend columns=1,
	legend cell align={left},
	legend style={row sep=3pt, at={(1,1)}, anchor = north east},
	]

    \addplot [line width=4pt, \myRed] table[x=X, y=Y, col sep=comma] 
    {fig/errRate_Xshort_data_uncompressed.txt} node[midway,below=0.05, fill=white, font=\Large\bf] {Uncompressed};
    
    \addplot+[line width=2pt, orange!80!black, mark=x, mark options={scale=4,fill=orange!80!black,solid}] table[x=X, y=Y, col sep=comma] 
    {fig/errRate_Xshort_data_opt.txt};
    
    \addplot+[line width=2pt, \myBlue, densely dotted, mark=*, mark options={scale=2,fill=\myBlue,solid}] table[x = X, y=Y, col sep=comma] {fig/errRate_Xshort_data_our.txt};
    
    \addplot+[line width=2pt, \myGreen, loosely dashed, mark=square, mark options={scale=2,fill=\myGreen,solid}] table[x=X, y=Y, col sep=comma] 
    {fig/errRate_Xshort_data_Q.txt};
    
    
    
    \legend{,Optimal compressor, Our compressor,Universal compressor}
\end{semilogyaxis}
\end{tikzpicture}
        \caption{
        Type-II error rates for $|\Xcal|=13$.
        }
        \label{fig:errRate-short}
\end{figure}

\begin{figure}[t]
    \centering
    \pgfplotsset{
    compat=newest,
    /pgfplots/legend image code/.code={%
        \draw[mark repeat=2,mark phase=2,#1] 
            plot coordinates {
                (0cm,0cm) 
                (0.6cm,0cm)
                (1.2cm,0cm)%
            };
    },
}

\begin{tikzpicture}
\begin{axis}[
	scale only axis,
	font=\Large,
	width=\textwidth,
	height=0.25\textwidth,
	xmajorgrids=true,
	ymajorgrids=true,
	xmin=1,
	xmax=7,
	ymin=0,
	ymax=1.5,
	title={Compression penalty $\Delta_\encF(R)$, $|\Xcal|=256$},
	xlabel={$\log M$},
	legend columns=1,
	legend cell align={left},
	legend style={row sep=3pt, at={(1,1)}, anchor = north east},
	]

    \addplot+[line width=2pt, \myBlue, densely dotted, mark=*, mark options={scale=2,fill=\myBlue,solid}] table[x = X, y=Y, col sep=comma] {fig/delta_data_our.txt};
    
    \addplot+[line width=2pt, \myGreen, loosely dashed, mark=square, mark options={scale=2,fill=\myGreen,solid}] table[x=X, y=Y, col sep=comma] 
    {fig/delta_data_Q.txt};

    \legend{Our compressor,Universal compressor}
\end{axis}
\end{tikzpicture}
    \caption{
    Compression penalty for $|\Xcal|=256$.
    }
    \label{fig:KL-results}
\end{figure}
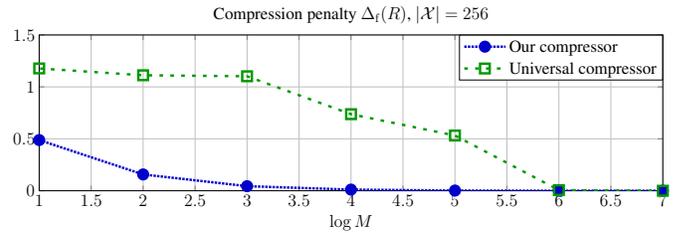

\begin{figure}[t]
	\centering
    \pgfplotsset{
    /pgfplots/legend image code/.code={%
        \draw[mark repeat=2,mark phase=2,#1] 
            plot coordinates {
                (0cm,0cm) 
                (0.6cm,0cm)
                (1.2cm,0cm)%
            };
    },
}

\begin{tikzpicture}
\begin{semilogyaxis}[
	scale only axis,
	font=\Large,
	width=\textwidth,
	height=0.3\textwidth,
	xmajorgrids=true,
	ymajorgrids=true,
	xmin=1,
	xmax=7,
    ymin=0.04,
	ymax=1,
	title={Type-II error rate, $|\Xcal|=256$, $n=5$, $\epsilon=0.05$},
	xlabel={$\log M$},
	legend columns=1,
	legend cell align={left},
	legend style={row sep=3pt, at={(1,1)}, anchor = north east},
	]

    \addplot [line width=4pt, \myRed] table[x=X, y=Y, col sep=comma] 
    {fig/errRate_data_uncompressed.txt} node[midway,below=0.1, fill=white, font=\Large\bf] {Uncompressed};
    
    
    \addplot+[line width=2pt, \myBlue, densely dotted, mark=*, mark options={scale=2,fill=\myBlue,solid}] table[x = X, y=Y, col sep=comma] {fig/errRate_data_our.txt};
    
    \addplot+[line width=2pt, \myGreen, loosely dashed, mark=square, mark options={scale=2,fill=\myGreen,solid}] table[x=X, y=Y, col sep=comma] 
    {fig/errRate_data_Q.txt};
    
    
    
    \legend{, Our compressor,Universal compressor}
\end{semilogyaxis}
\end{tikzpicture}
	\caption{
	Type-II error rates for $|\Xcal|=256$.
	}
	\label{fig:errRate}
\end{figure}
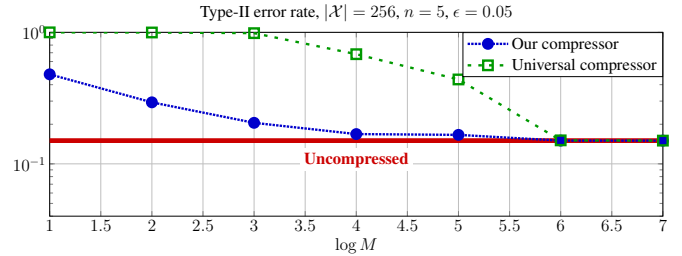

In Fig.~\ref{fig:distributions}, \ref{fig:KL-results-short} and~\ref{fig:errRate-short} we consider a source alphabet of size $|\Xcal|=13$; the parameters of the two hypotheses are $s_0=0.4$, $s_1=0.6$.
On the other hand, in Fig.~\ref{fig:KL-results} and~\ref{fig:errRate} we consider a larger source alphabet of size $|\Xcal|=256$;
the parameters of the two hypotheses are $s_0=0.48$, $s_1=0.52$.
We note that for this larger source alphabet, it is no longer computationally feasible to determine the optimal compressor.

Fig.~\ref{fig:distributions} illustrates the resulting KL-greedy compressor, the universal compressor, and the compressed distributions for $M=4$. 
As discussed in Section~\ref{sec:compress-hyptest}, our KL-greedy compressor seeks to minimize the KL distance between the posteriors over the groups; we also point out that this induces a partition on $\Xcal$ that divides the source alphabet in regions where one of the hypothesis is more likely than the other.
This pattern is also visible in the compressed distributions, since the two hypotheses exhibit divergent distributions (large KL distance). 
On the other hand, the universal compressor aims to make the two compressed distributions as uniform as possible.
Clearly, as we discussed in Section~\ref{sec:compress-hyptest}, the larger KL divergence between the compressed distributions, the better for the hypothesis testing task.

Fig.~\ref{fig:KL-results-short} and~\ref{fig:KL-results} show the compression penalty as a function of the compression rate $M$.
The former also shows the performance of the optimal compressor, since it can be computed in reasonable time for a small source alphabet; in this case, we can see that our compressor performs close to the optimal.
In both cases, our compressor outperforms the universal compressor, and it quickly achieves zero penalty, i.e., the KL distance of the compressed distributions is close to the uncompressed one as $M$ increases.

Fig.~\ref{fig:errRate-short} and~\ref{fig:errRate} show the empirical type-II error rate as a function of the compression rate $M$.
The former also shows the performance of the optimal compressor: our compressor performance overlaps with the optimal compressor.
For both the small and the large alphabet scenarios, our compressor outperforms the universal compressor, and it quickly achieves an error rate close to the uncompressed setting as $M$ increases.

\section{Conclusion}
\label{sec:conclusion}

In this paper, we have analyzed one-shot lossy source coding for task-oriented communications.
We have provided a problem formulation where the transmitter has to compress data coming from one of two distribution, and the goal is to carry out hypothesis testing at the receiver side. 
We have proposed a greedy compression function which can be determined in polynomial time and aims to preserve the \emph{useful} information for hypothesis testing at the receiver. 
Namely, our scheme is designed to minimize the gap between the KL divergences at the source and after compression.
Our experimental results show that our compressor outperforms classical universal compression schemes and achieves error rate comparable to the uncompressed case even for low rates.

\bibliographystyle{bibliography/IEEEtran}
\bibliography{bibliography/IEEEabrv,bibliography/references}

\end{document}